\batchmode
\makeatletter
\def\input@path{{\string"H:/My Documents/XiaojingFiles/BSDE&StochasticControl/ProjectBarrierOption/Report/Reportpostfinal/\string"/}}
\makeatother
\documentclass[english]{article}
\usepackage[T1]{fontenc}
\usepackage[utf8]{luainputenc}
\usepackage{array}
\usepackage{verbatim}
\usepackage{float}
\usepackage{amsmath}
\usepackage{amsthm}
\usepackage{esint}

\makeatletter

\providecommand{\tabularnewline}{\\}

\theoremstyle{plain}
\newtheorem{thm}{\protect\theoremname}
  \theoremstyle{remark}
  \newtheorem{rem}{\protect\remarkname}
  \theoremstyle{plain}
  \newtheorem{lem}{\protect\lemmaname}

\makeatother

\usepackage{babel}
  \providecommand{\lemmaname}{Lemma}
  \providecommand{\remarkname}{Remark}
\providecommand{\theoremname}{Theorem}

\begin{document}

\title{Deep-learning based numerical BSDE method for barrier options}

\author{Bing Yu\thanks{Corporate Model Risk, Wells Fargo, bing.yu@wellsfargo.com},
Xiaojing Xing\thanks{Corporate Model Risk, Wells Fargo}, Agus Sudjianto\thanks{Corporate Model Risk, Wells Fargo} }
\maketitle
\begin{abstract}
As is known, an option price is a solution to a certain partial differential
equation (PDE) with terminal conditions (payoff functions). There
is a close association between the solution of PDE and the solution
of a backward stochastic differential equation (BSDE). We can either
solve the PDE to obtain option prices or solve its associated BSDE.
Recently a deep learning technique has been applied to solve option
prices using the BSDE approach. In this approach, deep learning is
used to learn some deterministic functions, which are used in solving
the BSDE with terminal conditions. In this paper, we extend the deep-learning
technique to solve a PDE with both terminal and boundary conditions.
In particular, we will employ the technique to solve barrier options
using Brownian motion bridges. 
\end{abstract}

\section{Introduction}

A barrier option is a type of derivative where the payoff depends
on whether the underlying asset has breached a predetermined barrier
price. For a simple barrier case, an analytical pricing formula is
available (see \cite{JohnHull}). Because barrier options have additional
conditions built in, they tend to have cheaper premiums than comparable
options without barriers. Therefore, if a trader believes the barrier
is unlikely to be reached, they may prefer to buy a knock-out barrier
option for a lower premium. There are different methods to solve option
prices, ranging from an analytical solution, solving PDE numerically,
and Monte Carlo simulations. Recently, a different approach using
machine learning has been proposed. 

Using machine learning to solve PDE was studied in \cite{EHJ}. In
this work, a new method was proposed for solving parabolic partial
differential equations with terminal conditions, which we will call
the standard framework hereafter. In this new method, the PDE is formulated
as a stochastic control problem through a Feymann-Kac formula. In
this formulation, a connection between PDE for option prices and BSDE
is made. The option price is obtained by solving the BSDE rather than
solving PDE. The solution to the BSDE is represented by two deterministic
functions. One innovation (shown in \cite{EHJ}) is the use of a neural
network and deep-learning technique to learn these deterministic functions.
The mathematical foundation of this approach is based on a Kolmogorov-Arnold
representation theorem. This theorem states that any continuous function
can be approximated by a finite composition of continuous functions
of a single variable. Cybenko (see\cite{Cybenko}) found that a feed-forward
neural network is a natural realization of the theorem and he provided
a concrete implementation using a sigmoid function. 

In addition to \cite{EHJ}, \cite{BEJ} extended the method to solve
fully non-linear PDE and second-order backward stochastic differential
equation. Other works related to this deep-learning method include
\cite{Raissi} and \cite{CMW}. In \cite{Raissi}, a different way
of simulating the processes in the forward-backward stochastic differential
equation (FBSDE) is proposed. Rather than using a neural network to
approximate the derivative of a PDE solution, the network is used
to directly approximate the PDE solution and the derivative is calculated
using automatic differentiation. A number of different choices for
building the neural network and learning structure and two new types
of structures are proposed in \cite{CMW} . These problems are in
the framework of a PDE with some terminal conditions. These PDE can
be solved by an equivalent BSDE. 

In the aforementioned standard framework, the PDE solved has no boundary
conditions. There are some works on PDE with free boundary conditions.
In these works, a BSDE is replaced by a reflected BSDE (RBSDE). A
penalty term is added to the loss function to take into account the
free boundary condition in order to solve the RBSDE. Again, machine
learning can be used in solving these problems. This approach is used
in \cite{FTT} to solve American options. Bermudan Swaptions is solved
by exercising the option at a boundary in \cite{WCSLS}. In our work,
we consider barrier options. We treat boundary conditions of barrier
options differently. Rather than using RBSDE with a penalty function
or exercise options at a boundary, we incorporated the boundary conditions
as terminal conditions. To our best knowledge, this approach has not
previously been done. 

In this paper, we organize as follows. In section 2, we present the
standard framework, which is designed for Cauchy problem. In section
3, we describe how we extend the standard framework to handle barrier
options, which corresponds to a Cauchy-Dirichlet problem. In section
4, we present numerical considerations and the results we obtained
from our experiments. Finally, we make some concluding remarks in
section 5.

\section{Basic method to solve BSDE by machine learning}

We briefly introduced the deep-learning-based numerical BSDE algorithm
proposed in \cite{EHJ}. We start from an FBSDE, which is first proposed
in \cite{PP}. 

\[
X_{t}=X_{0}+\int_{0}^{t}b_{s}(X_{s})ds+\int_{0}^{t}\sigma_{s}(X_{s})dW_{s}
\]

\[
Y_{t}=h(X_{T})+\int_{t}^{T}f_{s}(X_{s},Y_{s},Z_{s})ds-\int_{t}^{T}Z_{s}dW_{s}
\]
Here, $\{W_{s}\}_{0<s<T}$ is a Brownian motion and $h(X_{T})$ is
the terminal condition. The pair $(Y,Z)_{0<t<T}$ solves the BSDE.
It is known that there exists a deterministic function $u=u(t,x)$
such that $Y_{t}=u(t,X_{t})$, $Z_{t}=\nabla u(t,X_{t})\sigma_{t}(X_{t})$
and $u(t,x)$ solves a quasi-linear PDE. For both the forward and
backward process, we can use Euler scheme to approximate: 
\begin{equation}
X_{t_{i+1}}\approx X_{t_{i}}+b_{t_{i}}(X_{t_{i}})(t_{i+1}-t_{i})+\sigma_{t_{i}}(X_{t_{i}})(W_{t_{i+1}}-W_{t_{i}})\label{eq:forward}
\end{equation}

\begin{equation}
Y_{t_{i+1}}\approx Y_{t_{i}}-f_{t_{i}}(X_{t_{i}},Y_{t_{i}},Z_{t_{i}})(t_{i+1}-t_{i})+Z_{t_{i}}(W_{t_{i+1}}-W_{t_{i}})\label{eq:Backward}
\end{equation}
Note that we have made the backward process to be forward; this is
a commonly used technique in treating FBSDEs. This set of equations
has the following interpretation on a given path: $X_{t_{i}}$ is
the underlying price; $Y_{t_{i}}$ is the option price and $Z_{t_{i}}$
is related to the delta at time $t_{i}$. 

In the deep-learning-based numerical BSDE algorithm, a neural network
structure is used to approximate the term $Z_{t_{i}}$ at each time
step with parameter $\theta$. Starting from an underlying price $X_{0}$
at time $0$ and an initial guess $Y_{0,}Z_{0}$, we use equations
(\ref{eq:forward}) and (\ref{eq:Backward}) to calculate $X_{t_{i+1}}$
and $Y_{t_{i+1}}$ at every time step until the terminal time $T$.
At terminal time, the loss is given by $l(\theta,Y_{0},Z_{0})=E[(Y_{T}-h(X_{T}))^{2}]$.
A stochastic gradient descent method is used to minimize the loss
function by iterating to the optimal value of $Y_{0}$,$Z_{0}$ and
$\theta$. Note that the success of this idea relies on the fact that
the neural network can approximate non-linear function $Z_{t}(X_{t})$
well; this is guaranteed by the work of Cybenko (see \cite{Cybenko}).

This framework solves a PDE with Cauchy conditions (standard framework).
However, there are certain types of options that will correspond to
a PDE with Cauchy-Dirichlet conditions. For example, a barrier option
is a Cauchy-Dirichlet PDE problem. In this paper, we would consider
extending the standard framework to handle this case.

\section{Extension of basic method to solve barrier option}

Barrier options are options where the payoff depends on whether the
underlying asset's price reaches a certain level during a certain
period of time. These barrier options can be classified as either
knock-out options or knock-in options. A knock-out option ceases to
exist when the underlying asset price reaches a certain barrier. A
knock-in option comes into existence only when the underlying asset
price reaches a barrier. An up-and-out call is a regular call option
that ceases to exist if the asset price reaches a barrier level, $B$,
that is higher than the current asset price. An up-and-in call is
a regular call option that comes into existence only if the barrier
is reached. The down-in and down-out option are similarly defined.
Under the Black-Scholes Framework, assuming constant coefficient,
it is not hard to derive an analytical solution for these kinds of
barrier options. Therefore, we will use an analytical solution as
the benchmark.

We would like to start from the most general form of a Cauchy-Dirichlet
problem. As is well known, the Feynman-Kac formula has provided a
way of translating the problem of a partial differential equation
into a probabilistic problem. A Dirichlet condition needs to be translated
in a probabilistic way by stopping the underlying diffusion process
as it exits from a domain. 
\begin{thm}
(See \cite{Gobet} Chapter 4). Let $W$ be a Brownian motion. Assume
process $X_{t}$ satisfies: 
\[
X_{t}=X_{0}+\int_{0}^{t}b_{s}(X_{s})ds+\int_{0}^{t}\sigma_{s}(X_{s})dW_{s}
\]
where $b$ and $\sigma$ satisfy some usual regularity and boundedness
conditions. Let $D$ be a bounded domain in real space and define
$\tau^{t,x}=inf\{s>t,X_{s}^{t,x}\notin D\}$ as the first exit time
from domain $D$ by process $X$ started from $(t,x)$. Assume the
boundary $\partial D$ is smooth. Assume functions $r,g:[0,T]\times\bar{D}\rightarrow R$
are continuous. Then the solution $u(t,x)$ of class $C^{1,2}$ of
the following PDE

\begin{equation}
\begin{cases}
\begin{array}{cc}
\partial_{t}u(t,x)+b(t,x)\partial_{x}u(t,x)+\frac{1}{2}\sigma(t,x)^{2}\partial_{xx}u(t,x)-r(t,x)u(t,x)=0, & t<T,x\in D\\
u(T,x)=g(T,x) & x\in\bar{D}\\
u(t,x)=g(t,x) & (t,x)\in[0,T]\times\partial D
\end{array}\end{cases}\label{eq:PDE-general}
\end{equation}
can be expressed by the following probabilistic representation

\[
u(t,x)=E[g(\tau^{t,x}\wedge T,X_{\tau^{t,x}\wedge T}^{t,x})e^{-\int_{t}^{\tau^{t,x}\wedge T}r(s,X_{s}^{t,x})ds}].
\]
\end{thm}
\begin{rem}
The domain $D$ is assumed to be bounded. In fact, it is enough for
the boundary to be compact.
\end{rem}
$u(t,x)$ is an average over all paths start at $(t,x).$ If the path
never exits the domain $D$ (i.e., $\tau^{t,x}\geq T$), we use the
value at terminal $g(T,X_{T}^{t,x})$. If the path exits the domain
(i.e., $\tau^{t,x}<T$), we use the boundary value at the point the
path exits, (i.e., $g(\tau^{t,x},X_{\tau^{t,x}}^{t,x})$). We can
write it this way:

\begin{eqnarray}
u(t,x)&=&E[g(T,X_{T}^{t,x})e^{-\int_{t}^{T}r(s,X_{s}^{t,x})ds}|\tau^{t,x}\geq T]P(\tau^{t,x}\geq T)\nonumber\\
&+&E[g(\tau^{t,x},X_{\tau^{t,x}}^{t,x})e^{-\int_{t}^{\tau^{t,x}}r(s,X_{s}^{t,x})ds}|\tau^{t,x}<T]P(\tau^{t,x}<T).
\end{eqnarray}

An up-out call barrier option is a special case of the above general
case with domain, terminal condition, and boundary condition specifically
defined. In up-out call option pricing, domain is $D=\{x<B\}$, terminal
condition is $g(T,x)=(x-K)^{+}1_{\{x<B\}}$, and boundary condition
is $g(t,x)=0$ when $x\geq B$. We can also assume constant drift
$b,$ volatility $\sigma$, and interest rate $r$ for simplicity.
Then, the second term in equation (4) can be dropped since the boundary
condition is zero; the probabilistic representation then becomes: 

\begin{eqnarray}
u(t,x)&=&E[g(T,X_{T}^{t,x})e^{-r(T-t)}|\tau^{t,x}\geq T]P(\tau^{t,x}\geq T)
\end{eqnarray}To calculate this probabilistic representation, we write it into conditional
expectation of the terminal value $X_{T}^{t,x}$. By using Lemma 2
shown in the Appendix, we have:

\begin{eqnarray}
u(t,x)&=&E\{g(T,X_T^{t,x})e^{-r(T-t)}P(\tau^{t,x}\ge T|X_T^{t,x})\}
\end{eqnarray}

The term $P(\tau^{t,x}\geq T|X_{T}^{t,x})$ is the probability that,
given the underlying value at time $t$ and $T$, the underlying process
never crosses the barrier in between. We can apply a commonly used
technique in dealing with the simulation of a stopped process - a
Brownian motion bridge. Given the start point and end point of a geometric
Brownian motion, the probability that the process never exceeds a
certain level in between can be explicitly given as stated in the
following Lemma.
\begin{lem}
(Brownian Motion Bridge). Assume $X_{t}$ follows $dX_{t}/X_{t}=bdt+\sigma dW_{t}$
and define 
\[
\xi(y)=exp[-\frac{2*ln(y/X_{t})*ln(y/X_{t+\Delta t})}{\sigma^{2}\Delta t}],
\]
then

(1) If $B>max(X_{t},X_{t+\Delta t}),$then $P[max_{t<s<t+\Delta t}X_{s}<B]=1-\xi(B)$

(2)If $B<min(X_{t},X_{t+\Delta t}),$then $P[min_{t<s<t+\Delta t}X_{s}>B]=1-\xi(B)$
\end{lem}
{} 

Taking the use of Lemma 1 and plugging in the terminal condition for
barrier option, the probabilistic representation becomes:

\begin{eqnarray} u(t,x)=E\{[(X_T^{t,x}-K)^+1_{\{X_T^{t,x}<B\}}e^{-r(T-t)}][1-e^{-\frac{2ln(B/x)*ln(B/X_T^{t,x})}{\sigma^2(T-t)}}] \}\end{eqnarray}The
only random variable in the expectation is the underlying at terminal
$X_{T}^{t,x}$ and, therefore, we can view this as a vanilla option
with a relatively complicated payoff. Then, we can use the standard
framework to solve this problem. The BSDE we use in our algorithm
is: 

\[
X_{T}=x+\int_{t}^{T}bX_{s}ds+\int_{t}^{T}\sigma X_{s}dW_{s}
\]

\[
Y_{t}=h(X_{T})-\int_{t}^{T}rY_{s}ds-\int_{t}^{T}Z_{s}dW_{s}
\]
where $h(X)=[(X-K)^{+}1_{\{X<B\}}][1-e^{-\frac{2ln(B/x)ln(B/X)}{\sigma^{2}(T-t)}}]$.
The Euler scheme and loss function described in the standard framework
(Section 2) can be applied. For completeness, the PDE corresponding
to this BSDE is:

\[
\begin{cases}
\begin{array}{c}
\partial_{t}u(t,x,x_{0})+b\partial_{x}u(t,x,x_{0})+\frac{1}{2}\sigma^{2}\partial_{xx}u(t,x,x_{0})-ru(t,x,x_{0})=0\\
u(T,x,x_{0})=h(x,x_{0})
\end{array}\end{cases}
\]
where $h(x,x_{0})=[(x-K)^{+}1_{\{x<B\}}][1-e^{-\frac{2ln(B/x_{0})ln(B/x)}{\sigma^{2}(T-t)}}]$.
Note that the value $u(t,x_{0},x_{0})$ is the price we want.

\section{Numerical consideration}

\subsection{Choosing basic neural network structure}

Similar to other works mentioned before, a neural network is used
to approximate the term $Z_{t_{i}}$at each time step. The structure
of the neural net can significantly influence the convergence. In
the original work in \cite{EHJ}, at each time step, the parameter
set $\theta_{t_{i}}$ are different (i.e., building N different neural
networks for the N time steps): 
\begin{equation}
Z_{t_{i}}=net^{\theta_{t_{i}}}(X_{t_{i}}),i=0...N.\label{eq:Weinannet}
\end{equation}
In our work, we choose to use a type of merged neural network, as
proposed by \cite{CMW}. Rather than building N neural networks, we
use one neural network for all the N time steps, so $\theta$ does
not need subscript $t_{i}$. We also need to add one additional dimension
to the underlying; that is, the neural network also needs to take
the time to maturity as an input for the merged neural network structure
to work: 
\begin{equation}
Z_{t_{i}}=net^{\theta}(X_{t_{i}},T-t_{i}),i=0...N.\label{eq:ournet}
\end{equation}
Finally, we choose to use Elu as the base function and use one hidden
layer with 20 neurons. Using more hidden layers or using more neurons
in a hidden layer does not apparently improve the convergence. Learning
rates are chosen differently in each test. All weights in the network
are initialized using normal or uniform random variable without pre-training.
We can achieve quite accurate results by using this structure in most
cases. However, for some isolated cases, a batch normalization technique
is needed to improve the convergence of the algorithm. We discuss
details on improving the convergence later.

To improve computational efficiency, we use variance discretization
rather than the commonly used time discretization. The number of time
steps needed will roughly depend on the total variance of the model;
in our case, it is the square of volatility times maturity time, $\sigma^{2}T$
. For cases where volatility is small, we can take a larger time step;
in cases where volatility is large, we take a smaller time step to
improve accuracy. As a rule of thumb, we use the following formula
to determine time step $timestep\#=max(80,\frac{\sigma^{2}T}{0.025})$.

To improve efficiency further, we use different learning rates. At
the beginning of the training, we use larger learning rates to search
for solutions that are close to the optimal point. When solutions
become relatively stable, we decrease the learning rate to zoom in
on the optimal solution. 

Finally, we used a stopping criteria to stop the training when a change
of the last 50 iteration average price is less than a certain small
amount. We also insist on a minimum iteration number that training
needs to run and a maximum iteration number where the program should
stop.

\subsection{Test result summary\label{sub:4.3-Test-Result-summary}}

\subsubsection{Test overview \label{sub:4.3.1-Test-overview}}

We ran the numerical experiment for a wide range of inputs for an
up-out call option. Namely, we vary the length of maturity, underlying
price, volatility, and barrier value. For simplicity, we set the interest
rate and drift to both be zero. In total, we tested 72 cases (see
Table \ref{tab:Barrier-Option-Info}). For all tests, we used the
tensorflow Adam optimizer for the stochastic gradient descent method.
We ran these 72 cases in three different settings, as shown in Table
\ref{tab:Setting-of-Test}. In the first setting, we used brutal force
of 8,000 iterations to ensure convergence. However, some isolated
results are still not converged. In the second setting, we introduced
batch normalization at every layer to improve convergence. In the
last setting, we only use batch normalization at the input layer.
In the following section, we will discuss the results in detail.

\begin{table}[H]
\caption{Barrier Option Information\label{tab:Barrier-Option-Info}}
\begin{tabular}{|c|c|c|c|c|c|c|c|}
\hline 
Option type &
Strike &
Underlying &
Maturity &
Rate &
Drift &
Volatility &
Barrier\tabularnewline
\hline 
\hline 
Up-out call &
23 &
17,22,27,32 &
0.5, 2 &
0 &
0 &
0.4,0.8,1.2 &
40,60,100\tabularnewline
\hline 
\end{tabular}

\end{table}

\begin{table}[H]
\caption{Setting of tests\label{tab:Setting-of-Test}}
\begin{tabular}{|c|c|c|c|}
\hline 
 &
Test 1 setting &
Test 2 setting &
Test 3 setting\tabularnewline
\hline 
\hline 
Layers &
\multicolumn{3}{c|}{3}\tabularnewline
\hline 
Neurons &
\multicolumn{3}{c|}{d+20}\tabularnewline
\hline 
Base function &
\multicolumn{3}{c|}{Elu}\tabularnewline
\hline 
Initial learning rate &
0.01 &
0.02 &
0.02\tabularnewline
\hline 
LR decay factor &
0.5 &
0.5 &
0.5\tabularnewline
\hline 
LR decay frequency &
1,500 &
500 &
1,000\tabularnewline
\hline 
Maximum iteration &
8,000 &
1,500 &
3,000\tabularnewline
\hline 
Stopping criteria &
No &
Price change<0.002 &
Price change<0.005\tabularnewline
\hline 
Minimum iteration &
8,000 &
750 &
1,500\tabularnewline
\hline 
Time step &
\multicolumn{3}{c|}{max(80,Var/0.025)}\tabularnewline
\hline 
Batch size &
\multicolumn{3}{c|}{512}\tabularnewline
\hline 
Batch normalization &
No &
At every layer &
At input layer\tabularnewline
\hline 
\end{tabular}

\end{table}

\subsubsection{Test result\label{sub:4.3.2-Wide-range-1=0000262}}

In this section, we present the results for 72 cases tested. In test
1, we ran 8,000 iterations. As Table \ref{tab:Wide-range-grid-stat}
indicates, on average the results differ from analytical solutions
by 4.5 cents and relative differences of 1.21\%. In approximately
50\% of cases, the pricing difference is less than a penny, while
there is a relative error of 0.4\%. Some isolated cases, shown in
Table \ref{tab:Result-of-isolated-case}, have large differences.
To improve convergence, we applied batch normalization technique.

Batch normalization technique is proposed in \cite{BN}. When we have
a neural network, the change in the input distribution at each layer
presents a problem because parameters need to continuously adapt to
a new distribution. This is a phenomenon known as covariate shift.
Eliminating internal covariate shift provides faster training and
batch normalization is the mechanism to do so. Batch normalization
accomplishes this via a normalization step that fixes the means and
variances of input distribution. We would like to note that, when
adding a batch normalization technique in tests to improve performance
in our test, we add additional trainable variables; these trainable
parameters within batch normalization are different among different
time steps (i.e., the network will still depend on $t_{i}$). 
\begin{equation}
Z_{t_{i}}=net^{\theta,\beta_{t_{i}}}(X_{t_{i}},T-t_{i}),i=0...N\label{eq:BN-Merge-Net}
\end{equation}
Here $\beta_{t_{i}}$ are the parameters introduced in batch normalization. 

Test 2, shown in Table \ref{tab:Wide-range-grid-stat}, are results
of applying batch normalization at every layer. Comparing the results
from test 1 and those from test 2, we can see significant improvement
in both absolute differences and relative differences when batch normalization
is applied. In addition, those isolated points where they failed to
converge in test 1 now converge nicely, with less than 1\% relative
differences, as shown in Table \ref{tab:Result-of-isolated-case}. 

When doing batch normalization, we can choose to apply it at every
layer or we can apply it only at the input layer. Whether we apply
batch normalization at each layer or just the input layer, the results
are similar, as shown in Table \ref{tab:Wide-range-grid-stat}. However,
computation efficiency is very different as we will show later. 

\begin{table}[H]
\caption{Grid test result - error statistics\label{tab:Wide-range-grid-stat}}
\begin{tabular}{|c|c|c|c|c|c|c|}
\hline 
Statistics &
Test 1 rel &
Test 2 rel  &
Test 3 rel  &
Test 1 abs &
Test 2 abs  &
Test 3 abs \tabularnewline
\hline 
\hline 
Average &
1.21\% &
0.57\% &
0.56\% &
0.0452 &
0.0099 &
0.0086\tabularnewline
\hline 
STD &
2.56\% &
0.54\% &
0.53\% &
0.1399 &
0.0116 &
0.0093\tabularnewline
\hline 
25\% quantile &
0.24\% &
0.21\% &
0.13\% &
0.0017 &
0.0017 &
0.0013\tabularnewline
\hline 
Median &
0.41\% &
0.39\% &
0.48\% &
0.0043 &
0.0059 &
0.0047\tabularnewline
\hline 
75\% quantile &
1.04\% &
0.70\% &
0.76\% &
0.0174 &
0.0133 &
0.0132\tabularnewline
\hline 
\end{tabular}

\end{table}

\begin{table}[H]
\caption{Result of isolated test cases\label{tab:Result-of-isolated-case}}
\begin{tabular}{|c|c|c|c|c|c|}
\hline 
Maturity &
Underlying &
Volatility &
Barrier &
Test 1 rel error &
Test 2 rel error\tabularnewline
\hline 
\hline 
0.5 &
17 &
1.2 &
100 &
18\% &
0.65\%\tabularnewline
\hline 
0.5 &
22 &
0.8 &
100 &
4.09\% &
0.34\%\tabularnewline
\hline 
0.5 &
27 &
0.8 &
100 &
6.76\% &
0.28\%\tabularnewline
\hline 
0.5 &
32 &
0.8 &
100 &
9.03\% &
0.39\%\tabularnewline
\hline 
2 &
22 &
0.4 &
100 &
4.02\% &
0.18\%\tabularnewline
\hline 
2 &
27 &
0.4 &
100 &
6.42\% &
0.15\%\tabularnewline
\hline 
\end{tabular}

\end{table}

\subsubsection{Improving the efficiency\label{sub:4.3.4-Tests-for-3=0000264}}

As we have mentioned before, we can apply batch normalization at every
layer or apply it only at the input layer. When we apply it at each
layer, we increase the trainable variable but less iterations are
required to achieve convergence. On the other hand, if we apply it
only at the input layer, we have less trainable parameters but more
iterations are needed to achieve convergence. However, the overall
time needed to achieve convergence is less in the latter case, as
shown in Table \ref{tab:Efficiency-Result}. The overall running time
for applying batch normalization at only the input layer is almost
three times faster. 

\begin{table}[H]
\caption{Efficiency Results\label{tab:Efficiency-Result}}
\begin{tabular}{|>{\raggedright}p{5cm}|c|c|}
\hline 
Indicator &
Test 2 result &
Test 3 result\tabularnewline
\hline 
\hline 
Average iteration step needed &
1000 &
1670\tabularnewline
\hline 
Time consumed per 200 training iterations &
15-20s &
4-5s\tabularnewline
\hline 
Approximate running time per case (including building time) &
200s &
75s\tabularnewline
\hline 
\end{tabular}

\end{table}

\section{Conclusion}

In this work, we solved a PDE with boundary conditions, using barrier
options as a concrete example. In this problem, the diffusion domain
is restricted by a barrier. By viewing the terminal condition probabilistically
(i.e., including a breaching probability of barrier), we are able
to recast this problem into the standard framework, namely a PDE with
terminal conditions. This PDE is solved by its equivalent BSDE using
a machine learning technique. We have completed extensive testing
using a wide range of market conditions and achieved good results
when comparing with known analytical results. In some isolated cases,
the batch normalization technique is needed to improve learning.

\section{Appendix}

For completeness, we present a technique Lemma that was used in section
3 for transition from equation (5) to (6).
\begin{lem}
Assume $X$ is a random variable, $A$ is an event. Then $E[f(X)|A]P(A)=E[f(X)P(A|X)]$
for any function $f(\cdot)$.\end{lem}
\begin{proof}
First starting from the left side, we have $E[f(X)|A]P(A)=E[1_{A}E[f(X)|A]]=E[E[f(X)1_{A}|A]]=E[f(X)1_{A}]$.
Then, starting from the right side, $E[f(X)P(A|X)]=E[f(X)E[1_{A}|X]]=E[E[f(X)1_{A}|X]]=E[f(X)1_{A}]$.
We arrive at same quantity from both side; thus, the statement is
proved.

\end{proof}

\end{document}